\documentclass[3p,12pt]{elsarticle}

\usepackage{graphicx}
\usepackage[utf8]{inputenc}  
\usepackage{amsmath}
\usepackage{amssymb}
\usepackage{amsthm}
\usepackage{latexsym}
\usepackage{float}
\usepackage{tabularx}
\usepackage{booktabs}

\allowdisplaybreaks

\journal{Science of Computer Programming}

\begin{document}

\newcommand{\Gp}{G^\prime}
\newcommand{\Tup}[2]{(#1,\,#2)}
\newcommand{\Tri}[3]{(#1,\,#2,\,#3)}
\newcommand{\Grm}[4]{(#1,\,#2,\,#3,\,#4)}
\newcommand{\Peg}[2]{#1[#2]}
\newcommand{\Pgg}[1]{\Peg{G}{#1}}
\newcommand{\Mat}[2]{#1\;\,#2\,}
\newcommand{\Matm}[3]{#1\;\,#2\;\,#3\,}
\newcommand{\Lp}{\stackrel{\mbox{\tiny{PEG}}}{\leadsto}}
\newcommand{\Li}{\stackrel{\mbox{\tiny{INC}}}{\leadsto}}
\newcommand{\Matg}[2]{\Mat{\Pgg{#1}}{#2}}
\newcommand{\Matgm}[3]{\Matm{\Pgg{#1}}{#2}{#3}}
\newcommand{\Matgmm}[2]{\Matm{\Pgg{#1}}{#2}{M}}
\newcommand{\Nothing}{{\tt fail}}
\newcommand{\Just}[1]{#1}
\newcommand{\fivespaces}{\;\;\;\;\;}
\newcommand{\tenspaces}{\fivespaces\fivespaces}
\newcommand{\twentyspaces}{\tenspaces\tenspaces}
\newcommand{\thirtyspaces}{\twentyspaces\tenspaces}
\newcommand{\fortyspaces}{\twentyspaces\twentyspaces}
\newcommand{\interf}{\fivespaces}
\newcommand{\mylabel}[1]{\ \mathbf{(#1)}}
\newcommand{\Con}[2]{#1\,#2}
\newcommand{\Choice}[2]{#1\,/\:#2}
\newcommand{\Pcfp}{P_{CFP}}
\newcommand{\Anyx}{X}
\newcommand{\Epsi}{\varepsilon}
\newcommand{\Ps}{p_S}
\newcommand{\Xp}{x^\prime}
\newcommand{\Yp}{y^\prime}

\def\verylongrightarrow{\hbox to 35pt{\rightarrowfill}}

\newcommand{\Pow}{\, \widehat{} \,\,}
\newcommand{\Matchf}[4]{\mathrm{match} \;\, #1\;#2\;#3\;#4}
\newcommand{\Matchl}[4]{\mathrm{match}_L \;\, #1\;#2\;#3\;#4}
\newcommand{\Justc}[2]{(\mbox{#1},\,#2)}
\newcommand{\Striple}[3]{(#1,\,#2,\,#3)}
\newcommand{\Sstate}[3]{\langle#1,\,#2,\,#3\rangle}
\newcommand{\Sstatel}[4]{\langle#1,\,#2,\,#3,\,#4\rangle}
\newcommand{\Sstatec}[5]{\langle#1,\,#2,\,#3,\,#4,\,#5\rangle}
\newcommand{\Sfail}[1]{\mbox{\bf Fail}{\langle#1\rangle}}
\newcommand{\Sstep}[4]{#1 & \stackrel{#2}{\verylongrightarrow} & #3 #4\\}
\newcommand{\Sstepf}[4]{#1 & \stackrel{#2}{\verylongrightarrow} & #3 #4}
\newcommand{\Ssteppp}[3]{#1 & \stackrel{\parbox{10pt}{\addvspace{5pt}}}{\verylongrightarrow} & #2 #3\\}
\newcommand{\Sstepppf}[3]{#1 & \stackrel{\parbox{10pt}{\addvspace{5pt}}}{\verylongrightarrow} & #2 #3}
\newcommand{\Sstepp}[4]{#1 \xrightarrow{#2} #3\\}
\newcommand{\Ssteppf}[4]{#1 \xrightarrow{#2} #3}
\newcommand{\Ia}[2]{\mbox{{\scriptsize {\tt #1} $#2$}}}
\newcommand{\Iaa}[3]{\mbox{{\scriptsize {\tt #1} $#2$ $#3$}}}
\newcommand{\In}[2]{\mbox{{{\tt #1} $#2$}}}
\newcommand{\Fail}{\mbox{\bf Fail}}
\newcommand{\Hd}[1]{\mbox{hd}(#1)}
\newcommand{\Sub}{{s}}
\newcommand{\Subb}{{s_1}}
\newcommand{\Subi}{\Sub [i]}
\newcommand{\Subbi}{\Subb [i]}
\newcommand{\MathN}[1]{\mbox{\emph{#1}}}
\newcommand{\Nat}{\mbox{\bf N}}
\newcommand{\myarrow}{\xrightarrow{\;\;\;\;\;}}
\newcommand{\myarrowstar}{\xrightarrow{\;\;*\;\;}}
\newcommand{\chmath}[1]{\mbox{`}#1\mbox{'}}

\newtheorem{definition}{Definition}[section]
\newtheorem{proposition}{Proposition}[section]
\newtheorem{corollary}[proposition]{Corollary}
\newtheorem{lemma}[proposition]{Lemma}

\begin{frontmatter}

\title{Left Recursion in Parsing Expression Grammars}

\author{S\'{e}rgio Medeiros}
\ead{sergio@ufs.br}
\address{Department of Computer Science -- UFS -- Aracaju -- Brazil} 

\author{Fabio Mascarenhas}
\ead{mascarenhas@ufrj.br}
\address{Department of Computer Science -- UFRJ -- Rio de Janeiro -- Brazil} 

\author{Roberto Ierusalimschy}
\ead{roberto@inf.puc-rio.br}
\address{Department of Computer Science -- PUC-Rio -- Rio de Janeiro -- Brazil} 

\begin{abstract}
Parsing Expression Grammars (PEGs) are a formalism
that can describe all deterministic context-free
languages through a set of rules that specify a top-down parser for
some language. PEGs are easy to use, and there are efficient
implementations of PEG libraries in several programming languages. 

A frequently missed feature of PEGs is left recursion, which is
commonly used in Context-Free Grammars (CFGs) to encode left-associative
operations. We present a simple conservative extension to the
semantics of PEGs that gives useful meaning to direct and indirect
left-recursive rules, and show that our extensions make it easy to
express left-recursive idioms from CFGs in PEGs, with similar
results. We prove the conservativeness of these extensions, and also
prove that they work with any left-recursive PEG.

PEGs can also be compiled to programs in a low-level {\em parsing
  machine}. We present an extension to the semantics of the operations
of this parsing machine that let it interpret left-recursive PEGs, and
prove that this extension is correct with regards to our semantics for
left-recursive PEGs.
\end{abstract}


\end{frontmatter}
     
\section{Introduction}
\label{sec:intro}

Parsing Expression Grammars (PEGs)~\cite{ford:peg} are a formalism for
describing a language's syntax, based on the TDPL and GTDPL formalisms
for top-down parsing with backtracking~\cite{birman:tdpl,aho:parsing}, 
and an alternative to the commonly used Context Free Grammars (CFGs). 
Unlike CFGs, PEGs are unambiguous
by construction, and their standard semantics is based on recognizing
instead of deriving strings. Furthermore, a PEG can be considered both the
specification of a language and the specification of a top-down parser
for that language. 

PEGs use the notion of {\em limited backtracking}: the parser, when faced with several
alternatives, tries them in a deterministic order (left to right),
discarding remaining alternatives after one of them succeeds.
This {\em ordered choice} characteristic makes PEGs unambiguous by
design, at the cost of making the language of a PEG harder to reason
about than the language of a CFG. The mindset of the author of a language specification
that wants to use PEGs should be closer to the mindset of the programmer of a hand-written
parser than the mindset of a grammar writer.

In comparison to CFGs, PEGs add the restriction that the order of alternatives in
a production matters (unlike the alternatives in the EBNF notation of CFGs, which can be reordered at will), 
but introduce a richer syntax, based on the syntax of extended
regular expressions. Extended regular expressions are a version of regular expressions popularized by pattern macthing
tools such as {\tt grep} and languages such as Perl. PEGs also add {\em syntactic predicates}, a form
of unrestricted lookahead where the parser checks whether the rest of
the input matches a parsing expression without consuming the input.

As discussed by Ford~\cite{ford:thesis}, with the help of syntactic
predicates it is easy to describe a scannerless PEG parser for a programming
language with reserved words and identifiers, a task that is not easy to
accomplish with CFGs. In a parser based on the CFG below, where we omitted
the definition of $Letter$, the non-terminal $I\!d$ could also match {\tt for},
a reserved word, because it is a valid identifier:
\begin{align*}
Id & \rightarrow Letter \; IdAux \\
IdAux & \rightarrow Letter \; IdAux \ | \ \Epsi \\
ResWord & \rightarrow if \ | \ for \ | \ while
\end{align*}

As CFGs do not have syntactic predicates, and context-free languages are not closed under complement,
 a separate lexer is necessary to distinguish reserved
words from identifiers. In case of PEGs, we can just use the not predicate
$!$, as below: 
\begin{align*}
Id & \leftarrow \;!(ResWord) \; Letter \; IdAux \\
IdAux & \leftarrow Letter \; IdAux \ / \ \Epsi \\
ResWord & \leftarrow if \ / \ for \ / \ while
\end{align*}

Moreover, syntactic predicates can also be used to better decide which
alternative of a choice should match, and address the limitations of ordered choice.
For example, we can use syntactic predicates to rewrite context-free grammars belonging to
some classes of grammars that have top-down predictive parsers as PEGs without changing the
gross structure of the grammar and the resulting parse trees~\cite{mascarenhas:cfgpeg}. 

The top-down parsing approach of PEGs means that they cannot handle
left recursion in grammar rules, as they would make the parser loop
forever. Left recursion can be detected structurally, so PEGs with
left-recursive rules can be simply rejected by PEG implementations
instead of leading to parsers that do not terminate, but the lack of
support for left recursion is a restriction on the expressiveness of
PEGs. The use of left recursion is a common idiom for expressing
language constructs in a grammar, and is present in published grammars
for programming languages; the use of left recursion can make
rewriting an existing grammar as a PEG a difficult task~\cite{roman:primitive}.
While left recursion can be eliminated from the grammar, the necessary
transformations can extensively change the shape of the resulting parse trees,
making posterior processing of these trees harder than it could be.

There are proposals for adding support for left recursion to PEGs, but
they either assume a particular PEG implementation
approach, {\em packrat parsing}~\cite{warth:left}, or support just
direct left recursion~\cite{tratt:left}. Packrat parsing~\cite{ford:packrat} is an
optimization of PEGs that uses memoization to guarantee linear time
behavior in the presence of backtracking and syntactic predicates, but
can be slower in practice~\cite{roman:aspects,mizushima}. Packrat parsing is a common
implementation approach for PEGs, but there are
others~\cite{dls:lpeg}. Indirect left recursion is
present in real grammars, and is difficult to rewrite the
grammar to eliminate it, because this can involve the
rewriting of many rules~\cite{roman:primitive}.

In this paper, we present a novel operational semantics for PEGs that
gives a well-defined and useful meaning for PEGs with left-recursive
rules. The semantics is given as a conservative extension of the
existing semantics, so PEGs that do not have left-recursive rules
continue having the same meaning as they had. It is also implementation agnostic, and
should be easily implementable on packrat implementations, plain
recursive descent implementations, and implementations based on a
parsing machine.

We also introduce {\em parse strings} as a possible semantic value
resulting from a PEG parsing some input, in parallel to the parse
trees of context-free grammars. We show that the parse strings that
left-recursive PEGs yield for the common left-recursive grammar
idioms are similar to the parse trees we get from bottom-up parsers
and left-recursive CFGs, so the use of left-recursive rules in PEGs
with our semantics should be intuitive for grammar writers.

A simple addition to our semantics lets us describe expression
grammars with multiple levels of operator precedence and associativity
more easily, in the style that users of LR parser generators are used
to.

Finally, PEGs can also be compiled to programs in a low-level {\em parsing
  machine}~\cite{dls:lpeg}. We present an extension to the semantics of the operations
of this parsing machine that let it interpret left-recursive PEGs, and
prove that this extension is correct with regards to our semantics for
left-recursive PEGs.

The rest of this paper is organized as follows: Section~\ref{sec:peg}
presents a brief introduction to PEGs and discusses the problem
of left recursion in PEGs; Section~\ref{sec:left} presents our
semantic extensions for PEGs with left-recursive
rules; Section~\ref{sec:prec} presents a simple addition that makes it
easier to specify operator precedence and associativity in expression
grammars; Section~\ref{sec:machine} presents left recursion as an
extension to a parsing machine for PEGs; Section~\ref{sec:related} reviews
some related work on PEGs and left recursion in more detail; finally,
Section~\ref{sec:con} presents our concluding remarks.

\section{Parsing Expression Grammars and Left Recursion}
\label{sec:peg}

Parsing Expression Grammars borrow the use of non-terminals and rules
(or productions) to express context-free recursion, although all
non-terminals in a PEG must have only one rule. The syntax of the
right side of the rules, the {\em parsing expressions}, is borrowed
from regular expressions and its extensions, in order to make it
easier to build parsers that parse directly from characters instead of
tokens from a previous lexical analysis step. The semantics of PEGs
come from backtracking top-down parsers, but in PEGs the backtracking
is local to each choice point.

Our presentation of PEGs is slightly different from
Ford's~\cite{ford:peg}. The style we use comes from
our previous work on PEGs~\cite{dls:lpeg,sblp:repeg},
and makes the exposition of
our extensions, and their behavior, easier to understand.
We define a PEG $G$ as a tuple $(V, T, P, p_S)$ where $V$ is the finite
set of non-terminals, $T$ is the alphabet (finite set of terminals),
$P$ is a function from $V$ to parsing expressions, and $p_S$ is the
{\em starting expression}, the one that the PEG matches. 
Function $P$ is commonly described through a set of rules
of the form $A \leftarrow p$, where $A \in V$ and $p$ is a
parsing expression.

Parsing expressions are the core of our
formalism, and they are defined inductively as the empty expression
$\varepsilon$, a terminal symbol $a$, a non-terminal symbol $A$, a
concatenation $p_1p_2$ of two parsing expressions $p_1$ and $p_2$, an
ordered choice $p_1 / p_2$ between two parsing expressions $p_1$ and
$p_2$, a repetition $p^*$ of a parsing expression $p$, or a
not-predicate $!p$ of a parsing expression $p$. We leave out
extensions such as the dot, character classes, strings, and the
and-predicate, as their addition is straightforward. 

\begin{figure}[p]
{\small
\begin{align*}
& \textbf{Empty String} \fivespaces
{\frac{}{\Matg{\Epsi}{x} \Lp \Tup{x}{\Epsi}}}
\mylabel{empty.1}  \fivespaces
\textbf{Terminal} \fivespaces
{\frac{}{\Matg{a}{ax} \Lp \Tup{x}{a}}} \mylabel{char.1} \\ \\
&  \tenspaces \tenspaces \tenspaces
{\frac{}{\Matg{b}{ax} \Lp \Nothing}} \mbox{ , } b \neq a \mylabel{char.2} 
\tenspaces
{\frac{}{\Matg{a}{\Epsi} \Lp \Nothing}} \mylabel{char.3}  \\ \\
& \textbf{Variable} \fivespaces
{\frac{\Matg{P(A)}{xy} \Lp \Tup{y}{\Xp}}
	{\Matg{A}{xy} \Lp \Tup{y}{A[\Xp]}}}    \mylabel{var.1}  
\tenspaces
{\frac{\Matg{P(A)}{x} \Lp \Nothing}
	{\Matg{A}{x} \Lp \Nothing}}    \mylabel{var.2}
\\ \\
& \textbf{Concatenation}
\fivespaces
{\frac{\Matg{p_1}{xyz} \Lp \Tup{yz}{\Xp} \interf \Matg{p_2}{yz}
    \Lp \Tup{z}{\Yp}}
	{\Matg{\Con{p_1}{p_2}}{xyz} \Lp \Tup{z}{\Xp \Yp}}}
      \mylabel{con.1} \\ \\
&
\fivespaces
{\frac{\Matg{p_1}{xy} \Lp \Tup{y}{\Xp} \interf \Matg{p_2}{y}
    \Lp \Nothing}
	{\Matg{\Con{p_1}{p_2}}{xy} \Lp \Nothing}}
      \mylabel{con.2}
 \fivespaces
{\frac{\Matg{p_1}{x} \Lp \Nothing}
	{\Matg{\Con{p_1}{p_2}}{x} \Lp \Nothing}} \mylabel{con.3} \\ \\
& \textbf{Choice} 
\fivespaces
{\frac{\Matg{p_1}{xy} \Lp \Tup{y}{\Xp}}
	{\Matg{\Choice{p_1}{p_2}}{xy} \Lp \Tup{y}{\Xp}}} \mylabel{ord.1} \fivespaces
{\frac{\Matg{p_1}{x} \Lp \Nothing \interf \Matg{p_2}{x} \Lp \Nothing}
	{\Matg{\Choice{p_1}{p_2}}{x} \Lp \Nothing}} \mylabel{ord.2} \\ \\ 
& {\frac{\Matg{p_1}{xy} \Lp \Nothing \interf \Matg{p_2}{xy} \Lp \Tup{y}{\Xp}}
	{\Matg{\Choice{p_1}{p_2}}{xy} \Lp \Tup{y}{\Xp}}} \mylabel{ord.3} \\ \\ 
& \textbf{Not Predicate} \tenspaces \fivespaces
{\frac{\Matg{p}{x} \Lp \Nothing} 
	{\Matg{!p}{x} \Lp \Tup{x}{\Epsi}}} \mylabel{not.1} \tenspaces
{\frac{\Matg{p}{xy} \Lp \Tup{y}{\Xp}}
	{\Matg{!p}{xy} \Lp \Nothing}} \mylabel{not.2} \\ \\
& \textbf{Repetition} \;\;
{\frac{\Matg{p}{x} \Lp \Nothing}
	{\Matg{p^*}{x} \Lp \Tup{x}{\Epsi}}} \mylabel{rep.1} \;\;
{\frac{\Matg{p}{xyz} \Lp \Tup{yz}{\Xp} \interf \Matg{p^*}{yz} \Lp \Tup{z}{\Yp}}
	{\Matg{p^*}{xyz} \Lp \Tup{z}{\Xp \Yp}}} \mylabel{rep.2} 
\end{align*}
\caption{Semantics of the $\Lp$ relation}
\label{fig:matchpeg}
}
\end{figure}

We define the semantics of PEGs via a relation $\Lp$ among a PEG,
a parsing expression, a subject, and a result. The notation
$\Matg{p}{xy} \Lp \Just{(y,x^\prime)}$ means that the expression $p$
matches the input $xy$, consuming the prefix $x$, while leaving $y$ and
yielding a {\em parse string} $x^\prime$ as the output, while resolving any non-terminals using
the rules of $G$. We use $\Matg{p}{xy} \Lp \Nothing$ to express
an unsuccessful match. The language of a PEG $G$ is defined as all
strings that $G$'s starting expression consumes, that is, the set $\{ x \in T^* \ | \
\Matg{p_s}{xy} \Lp \Just{(y, x^\prime)} \}$.

Figure~\ref{fig:matchpeg} presents the
definition of $\Lp$ using natural semantics~\cite{kahn,winskel}, as a set of inference
rules. Intuitively, $\varepsilon$ just succeeds and leaves the subject
unaffected; $a$ matches and consumes itself, or fails; $A$
tries to match the expression $P(A)$; $p_1p_2$ tries to match $p_1$,
and if it succeeds tries to match $p_2$ on the part of the subject that
$p_1$ did not consume; $p_1 / p_2$ tries to match $p_1$, and if it
fails tries to match $p_2$; $p^*$ repeatedly tries to match $p$ until
it fails, thus consuming as much of the subject as it can; finally, $!p$
tries to match $p$ and fails if $p$ succeeds and succeeds if $p$
fails, in any case leaving the subject unaffected. It is easy to see
that the result of a match is either failure or a suffix of the
subject (not a proper suffix, as the expression may succeed without
consuming anything).

Context-Free Grammars have the notion of a {\em parse tree}, a
graphical representation of the structure that a valid subject has,
according to the grammar. The proof trees of our semantics can have a
similar role, but they have extra information that can obscure the
desired structure. This problem will be exacerbated in the proof trees that our
rules for left-recursion yield, and is the reason we introduce parse strings
 to our formalism. A parse string is roughly a linearization of a
 parse tree, and shows which non-terminals have been used in the
 process of matching a given subject; rule {\bf var.1} brackets the parse
string yielded by the right side of a non-terminal's rule,
and tags it with the left side. The brackets and the tag are part of the parse string,
and are assumed to not be in $T$, the PEG's input alphabet.
Having the result of a parse be an actual tree and having arbitrary semantic actions are
 straightforward extensions.

When using PEGs for parsing it is important to guarantee that a given
grammar will either yield a successful result or $\Nothing$ for every
subject, so parsing always terminates. Grammars where this is true are
{\em complete}~\cite{ford:peg}. In order to guarantee completeness, it
is sufficient to check for the absence of
direct or indirect {\em left recursion}, a property that can be
checked structurally using the {\em well-formed} 
predicate from Ford~\cite{ford:peg} (abbreviated {\em WF}). 

Inductively, empty expressions and symbol expressions are always
well-formed; a non-terminal is well-formed if it has a production and
it is well-formed; a choice is well-formed if the alternatives are
well-formed; a not predicate is well-formed if the expression it uses is
well-formed; a repetition is well-formed if the expression it repeats
is well-formed and cannot succeed without consuming input; finally, a concatenation is
well-formed if either its first expression is well-formed and cannot
succeed without consuming input or both of its expressions are
well-formed. 

A grammar is well-formed if its non-terminals and starting expression are all
well-formed. The test of whether an expression cannot succeed while
not consuming input is also computable from the structure of the
expression and its grammar from an inductive
definition~\cite{ford:peg}. The rule for well-formedness of
repetitions just derives from writing a repetition $p^*$ as a
recursion $A \leftarrow pA \ /\ \varepsilon$, so a non-well-formed
repetition is just a special case of a left-recursive rule.

Left recursion is not a problem in the popular bottom-up parsing
approaches, and is a natural way to express several common parsing
idioms. Expressing repetition using left recursion in a CFG yields a
left-associative parse tree, which is often desirable when parsing
programming languages, either because operations have to be
left-associative or because left-associativity is more
efficient in bottom-up parsers~\cite{grune}. For example, the following is a simple
left-associative CFG for additive expressions, written in EBNF notation:
\begin{align*}
E & \rightarrow E + T \ | \ E - T \ |\ T \\
T & \rightarrow n \ | \ \boldsymbol{(}E \boldsymbol{)}
\end{align*}

Rewriting the above grammar as a PEG, by replacing $|$ with the ordered
choice operator, yields a non-well-formed PEG that does not have a proof tree for any
subject. We can rewrite the grammar to eliminate the left recursion,
giving the following CFG, again in EBNF (the curly brackets are
metasymbols of EBNF notation, and express zero-or-more
repetition, while the parentheses are terminals):
\begin{align*}
E &\rightarrow T \{E^\prime\} \\
T &\rightarrow n \ | \ \boldsymbol{(}E\boldsymbol{)} \\
E^\prime& \rightarrow + T \ | \ - T
\end{align*}

This is a simple transformation, but it yields a different parse tree,
and obscures the intentions of the grammar writer, even though it is
possible to transform the parse tree of the non-left-recursive grammar
into the left-associative parse tree of the left-recursive
grammar. But at least we can straightforwardly express the non-left-recursive grammar with
the following PEG:
\begin{align*}
E & \leftarrow T\ E^{\prime *} \\
T & \leftarrow n \ / \ \boldsymbol{(}E \boldsymbol{)} \\
E^\prime & \leftarrow + T \ / \ - T
\end{align*}

Indirect left recursion is harder to eliminate, and its elimination
changes the structure of the grammar and the resulting trees even
more. For example, the following indirectly left-recursive CFG denotes a
very simplified grammar for l-values in a language with variables,
first-class functions, and records (where $x$ stands for identifiers and $n$
for expressions):
\begin{align*}
L & \rightarrow P.x \ | \ x \\
P & \rightarrow P\boldsymbol{(}n\boldsymbol{)} \ | \ L 
\end{align*}

This grammar generates $x$ and $x$ followed by any number of $(n)$ or
$.x$, as long as it ends with $.x$. An l-value is a prefix expression
followed by a field access, or a single variable, and a prefix
expression is a prefix expression followed by an operand, denoting a
function call, or a valid l-value. In the parse trees for this grammar each $(n)$ or $.x$ associates
to the left. 

Writing a PEG that parses the same language is difficult. We can eliminate the indirect left
recursion on $L$ by substitution inside $P$, getting $P \rightarrow
P\boldsymbol{(}n\boldsymbol{)} \ | \ P.x \ | \ x$, and then eliminate
the direct left recursion on $P$ to get the following CFG:
\begin{align*}
L & \rightarrow P.x \ | \ x \\
P & \rightarrow x \{P^\prime \} \\
P^\prime & \rightarrow \boldsymbol{(}n\boldsymbol{)} \ | \ .x
\end{align*}

But a direct translation of this CFG to a PEG will not work because PEG
repetition is greedy; the repetition on $P^\prime$ will consume the
last $.x$ of the l-value, and the first alternative of $L$ will always
fail. One possible solution is to not use the $P$ non-terminal in $L$,
and encode l-values directly with the following PEG (the bolded
parentheses are terminals, the non-bolded parentheses are metasymbols
of PEGs that mean grouping):
\begin{align*}
L & \leftarrow x \ S^*
\\
S & \leftarrow (\ \boldsymbol{(}n\boldsymbol{)}\ )^* .x
\end{align*}

The above uses of left recursion are common in published grammars,
with more complex versions (involving more rules and a deeper level of
indirection) appearing in the grammars in the specifications of
Java~\cite{spec:java} and Lua~\cite{spec:lua}. Having a straightforward
way of expressing these in a PEG would make the process of translating
a grammar specification from an EBNF CFG to a PEG easier and less error-prone.

In the next section, we will propose a semantic extension to the PEG
formalism that will give meaningful proof trees to left-recursive
grammars. In particular, we want to have the straightforward
translation of common left-recursive idioms such as left-associative
expressions to yield parse strings that are similar in structure to
parse trees of the original CFGs.

\section{Bounded Left Recursion}
\label{sec:left}

Intuitively, {\em bounded left recursion} is a use of a non-terminal
where we limit the number of left-recursive uses it may have. This is
the basis of our extension for supporting left recursion in PEGs. We use the notation $A^n$ to mean a
non-terminal where we can have less than $n$ left-recursive uses,
with $A^0$ being an expression that always fails. Any left-recursive
use of $A^n$ will use $A^{n-1}$, any left-recursive use of $A^{n-1}$
will use $A^{n-2}$, and so on, with $A^1$ using $A^0$ for any
left-recursive use, so left recursion will fail for $A^1$.

For the left-recursive definition $E \leftarrow E+n \ / \ n$ we have
the following progression, where we write expressions equivalent to
$E^n$ on the right side:
\begin{align*}
E^0 & \leftarrow \Nothing \\
E^1 & \leftarrow E^0 + n \ / \ n \ = \ \bot +n \ / \ n \ = \ n \\
E^2 & \leftarrow E^1 + n \ / \ n \ = \ n+n \ / \ n \\
E^3 & \leftarrow E^2 + n \ / \ n \ = \ (n+n \ / \ n)+n \ / \ n \\
& \vdots \\
E^n & \leftarrow E^{n-1} + n \ / \ n
\end{align*}

It would be natural to expect that increasing the bound will
eventually reach a fixed point with respect to a given subject, but
the behavior of the ordered choice operator breaks this expectation.
For example, with a subject {\tt n+n} and the previous PEG, $E^2$ will
match the whole subject, while $E^3$ will match just the first
{\tt n}. Table~\ref{tab:matching2} summarizes the 
results of trying to match some subjects against $E$ with different
left-recursive bounds (they show the
suffix that remains, not the matched prefix). 

\newcolumntype{R}{>{\raggedleft\arraybackslash}X}

\begin{table}[t]
\begin{tabularx}{\columnwidth}{ *{8}{R} } 
\toprule
Subject & $E^0$     & $E^1$      & $E^2$    & $E^3$    & $E^4$    & $E^5$    & $E^6$    \\ 
\midrule
{\tt  n}       & \Nothing  & $\Epsi$    & $\Epsi$  & $\Epsi$  & $\Epsi$  & $\Epsi$  & $\Epsi$    \\ 
{\tt  n+n}     & \Nothing  & {\tt +n}          & $\Epsi$  & {\tt +n}
& $\Epsi$  & {\tt +n}        & $\Epsi$     \\ 
{\tt  n+n+n}    & \Nothing  & {\tt +n+n}         & {\tt +n}        &
$\Epsi$  & {\tt +n+n}       & {\tt +n}        & $\Epsi$     \\ 
\bottomrule 
\end{tabularx}
\vspace{0.1in}
\caption{Matching $E$ with different bounds}
\label{tab:matching2}
\end{table}

The fact that increasing the bound can lead to matching a smaller
prefix means we have to pick the bound carefully if we wish to match
as much of the subject as possible. Fortunately, it is sufficient to
increase the bound until the size of the matched prefix stops
increasing.  In the above example, we would pick $1$ as the bound for
{\tt  n}, $2$ as the bound for {\tt  n+n}, and $3$ as the bound for {\tt  n+n+n}.

When the bound of a non-terminal $A$ is $1$ we are effectively
prohibiting a match via any left-recursive path, as all left-recursive
uses of $A$ will fail. $A^{n+1}$ uses $A^n$ on all its left-recursive paths, so 
if $A^n$ matches a prefix of length $k$, $A^{n+1}$ matching a prefix
of length $k$ or less means that either there is nothing to do
after matching $A^n$ (the grammar is cyclic), in which case it is
pointless to increase the bound after $A^n$, or all paths starting
with $A^n$ failed, and the match actually used a non-left-recursive
path, so $A^{n+1}$ is equivalent with $A^1$. Either option means
that $n$ is the bound that makes $A$ match the longest prefix of the subject.

We can easily see this dynamic in the $E \leftarrow E+n \ / \ n$
example. To match $E^{n+1}$ we have to match $E^n+n \ / n$. Assume
$E^n$ matches a prefix $x$ of the input. We then try to match the
rest of the input with $+n$, if this succeeds we will have matched
$x{\tt +n}$, a prefix bigger than $x$. If this fails we will have matched just {\tt  n},
which is the same prefix matched by $E^1$.

Indirect, and even mutual, left recursion is not a problem, as the
bounds are on left-recursive {\em uses} of a non-terminal, which are a property
of the proof tree, and not of the structure of the PEG. The bounds on two mutually
recursive non-terminals $A$ and $B$ will depend on which non-terminal
is being matched first, if it is $A$ then the bound of $A$ is fixed
while varying the bound of $B$, and vice-versa. A particular case of
mutual left recursion is when a non-terminal is both left and
right-recursive, such as $E \leftarrow E + E / n$. In our semantics,
$E^n$ will match $E^{n-1} + E / n$, where the right-recursive use of
$E$ will have its own bound. Later in this section we will elaborate
on the behavior of both kinds of mutual recursion.

In order to extend the semantics of PEGs with bounded left recursion,
we will show a conservative extension of the rules in
Figure~\ref{fig:matchpeg}, with new rules for left-recursive
non-terminals. For non-left-recursive non-terminals we will still use
rules {\bf var.1} and {\bf var.2}, although we will later prove that
this is unnecessary, and the new rules for non-terminals can replace
the current ones. The basic idea of the extension is
to use $A^1$ when matching a left-recursive non-terminal $A$
for the first time, and then try to increase the bound, while using
a memoization table $\mathcal{L}$ to keep the result of the current
bound. We use a different relation, with its own inference rules, for
this iterative process of increasing the bound.

\begin{figure}[t]
{\small
\begin{align*}
& \textbf{Left-Recursive Variable} \\
& 
{\frac{
\begin{array}{c}
\;\;\;\;
(A,xyz) \notin \mathcal{L} \interf
    \Matgm{P(A)}{xyz}{\mathcal{L}[(A,xyz) \mapsto \Nothing]} \Lp
    (yz,x^\prime)  \;\;\;\; \\ \;\;\;\;
				\Matgm{P(A)}{xyz}{\mathcal{L}[(A,xyz)
                                  \mapsto (yz,x^\prime)]} \Li
                                (z,(xy)^\prime) \;\;\;\;
\end{array}
}
	{\Matgm{A}{xyz}{\mathcal{L}} \Lp (z,A[(xy)^\prime])}} \mylabel{lvar.1} \\\\
& 
 \fivespaces {\frac{(A,x) \notin \mathcal{L} \;\;\;
    \Matgm{P(A)}{x}{\mathcal{L}[(A,x) \mapsto \Nothing]} \Lp \Nothing}
	{\Matgm{A}{x}{\mathcal{L}} \Lp \Nothing}} \mylabel{lvar.2} \\
      \\
&
      \fivespaces {\frac{\mathcal{L}(A,xy) = \Nothing}
	{\Matgm{A}{xy}{\mathcal{L}} \Lp \Nothing}} \mylabel{lvar.3}
      \interf
       {\frac{\mathcal{L}(A,xy) = (y,x^\prime)}
	{\Matgm{A}{xy}{\mathcal{L}} \Lp (y,A[x^\prime])}} \mylabel{lvar.4}
\\ \\
& \textbf{Increase Bound} \\
& 
{\frac{
\begin{array}{c}
\;\;\;\; \Matgm{P(A)}{xyzw}{\mathcal{L}[(A, xyzw) \mapsto (yzw,
      x^\prime)]} \Lp (zw, (xy)^\prime) \;\;\;\; \\  
\;\;\;\;	   \Matgm{P(A)}{xyzw}{\mathcal{L}[(A,xyzw) \mapsto (zw,
             (xy)^\prime)]} \Li (w, (xyz)^\prime) \;\;\;\;
\end{array}
}
	{\Matgm{P(A)}{xyzw}{\mathcal{L}[(A, xyzw) \mapsto (yzw,
            x^\prime)]} \Li (w, (xyz)^\prime)}} \mathrm{, where } \;y \neq \Epsi
	 \mylabel{inc.1} \\ \\
& 
{\frac{\Matgm{P(A)}{x}{\mathcal{L}} \Lp \Nothing}
	{\Matgm{P(A)}{x}{\mathcal{L}} \Li \mathcal{L}(A,x)}} \mylabel{inc.2}  
\;
{\frac{\Matgm{P(A)}{xyz}{\mathcal{L}[(A, xyz) \mapsto (z, (xy)^\prime)]} \Lp (yz, x^\prime)}
	{\Matgm{P(A)}{xyz}{\mathcal{L}[(A, xyz) \mapsto (z,
            (xy)^\prime)]} \Li (z, (xy)^\prime)}}
	 \mylabel{inc.3} 
\end{align*}
\caption{Semantics for PEGs with left-recursive non-terminals}
\label{fig:leftrecmemo}
}
\end{figure}

Figure~\ref{fig:leftrecmemo} presents the new rules. We give the behavior of
the memoization table $\mathcal{L}$ in the usual substitution style, where
$\mathcal{L} [(A,x) \mapsto \Anyx](B,y) = \mathcal{L}(B,y)$ if $B \neq A$ or $y \neq x$
and $\mathcal{L} [(A,x) \mapsto \Anyx](A,x) = \Anyx$ otherwise. All of
the rules in Figure~\ref{fig:matchpeg} just ignore this extra
parameter of relation $\Lp$. We also have rules for the new relation
$\Li$, responsible for the iterative process of finding the correct bound for a given
left-recursive use of a non-terminal.

Rules {\bf lvar.1} and {\bf lvar.2} apply the first time a
left-recursive non-terminal is used with a given subject, and they try
to match $A^1$ by trying to match the production of $A$ using $\Nothing$
for any left-recursive use of $A$ (those uses will fail through rule {\bf
  lvar.3}). If $A^1$ fails we do not try bigger bounds (rule {\bf
  lvar.2}), but if $A^1$ succeeds we store the result in $\mathcal{L}$ and try to
find a bigger bound (rule {\bf lvar.1}). Rule {\bf
  lvar.4} is used for left-recursive invocations of $A^n$ in the
process of matching $A^{n+1}$.

Relation $\Li$ tries to find the bound where $A$ matches the longest prefix. Which
rule applies depends on whether matching the production of $A$ using
the memoized value for the current bound leads to a longer match or
not; rule {\bf inc.1} covers the first case, where we use relation
$\Li$ again to continue increasing the bound after updating
$\mathcal{L}$. Rules {\bf inc.2} and {\bf inc.3} cover the second
case, where the current bound is the correct one and we just return
its result.

Let us walk through an example, again using $E \leftarrow E+n \ / \
n$ as our PEG, with {\tt n+n+n} as the subject. When first matching $E$ against
{\tt n+n+n} we have $(E,{\tt n+n+n}) \not\in \mathcal{L}$, as $\mathcal{L}$ is
initially empty, so we have to match $E+n \ / \ n$ against
{\tt n+n+n} with $\mathcal{L} = \{ (E,{\tt n+n+n}) \mapsto \Nothing \}$. We now
have to match $E+n$ against {\tt n+n+n}, which means matching $E$ again,
but now we use rule {\bf lvar.3}. The first alternative, $E+n$, fails,
and we have $\Matgm{E+n \ / n}{{\tt n+n+n}}{\{ (E,
  {\tt n+n+n}) \mapsto \Nothing\}} \Lp ({\tt +n+n},{\tt n})$ using the second
alternative, $n$, and rule {\bf ord.3}. 

In order to finish rule {\bf lvar.1} and the initial match we have to try
to increase the bound through relation $\Li$ with $\mathcal{L} = \{
(E,{\tt n+n+n}) \mapsto ({\tt +n+n},{\tt n}) \}$. 
This means we must try to match $E+n \ / \ n$ against {\tt n+n+n}
again, using the new $\mathcal{L}$. When we try the first alternative and match $E$ with
{\tt n+n+n} the result will be $({\tt +n+n},E[{\tt n}])$ via {\bf lvar.4}, and we can
then use {\bf con.1} to match $E+n$ yielding $({\tt +n},E[{\tt n}]{\tt
  +n})$. We have successfully increased the bound, and are in
rule {\bf inc.1}, with $x = {\tt n}$, $y = {\tt +n}$, and $zw = {\tt +n}$. 

In order to finish rule {\bf inc.1} we have to try to increase the
bound again using relation $\Li$, now with $\mathcal{L} = \{
(E,{\tt n+n+n}) \mapsto ({\tt +n},E[{\tt n}]{\tt +n}) \}$. We try to
match $P(E)$ again with this new $\mathcal{L}$, and this yields $(\varepsilon,
E[E[{\tt n}]{\tt +n}]{\tt +n})$ via {\bf lvar.4}, {\bf con.1}, and
{\bf ord.1}. We have successfully increased the bound and are using
rule {\bf inc.1} again, with $x = {\tt n+n}$, $y = {\tt +n}$, and $zw
= \varepsilon$. 

We are in rule {\bf inc.1}, and have to try to increase the bound a
third time using the relation $\Li$, with $\mathcal{L} = \{
(E,{\tt n+n+n}) \mapsto (\varepsilon,E[E[{\tt n}]{\tt +n}]{\tt +n})
\}$.  We have to match $E+n \ / n$ against {\tt n+n+n} again, using
this $\mathcal{L}$.
 In the first alternative $E$ matches and yields $(\varepsilon, E[E[E[{\tt n}]{\tt +n}]{\tt +n}])$
via {\bf lvar.4}, but the first alternative itself fails via {\bf
  con.2}. We then have to match $E+n \ / \ n$ against {\tt n+n+n} using {\bf
  ord.2}, yielding $({\tt +n+n},{\tt n})$. The attempt to increase the
bound for the third time failed (we are back to the same result we had
when $\mathcal{L} = \{ (A,{\tt n+n+n}) \mapsto \Nothing \}$), and we
use rule {\bf inc.3} once and rule {\bf inc.1} twice to propagate
$(\varepsilon,E[E[{\tt n}]{\tt +n}]{\tt +n})$ back to rule {\bf lvar.1}, and use this
rule to get the final result, $\Matgm{E}{{\tt n+n+n}}{\{\}} \Lp (\varepsilon,E[E[E[{\tt n}]{\tt +n}]{\tt +n}])$.

We can see that the parse string $E[E[E[{\tt n}]{\tt +n}]{\tt +n}]$ implies left-associativity in the $+$
operations, as intended by the use of a left-recursive rule.

More complex grammars, that encode different precedences and
associativities, behave as expected. For example, the following grammar
has a right-associative $+$ with a left-associative $-$:
\begin{align*}
E & \leftarrow M + E \ / \  M \\
M & \leftarrow M - n \ / \ n
\end{align*}

Matching $E$ with {\tt n+n+n} yields $E[M[{\tt n}]{\tt +}E[M[{\tt n}]{\tt +}E[M[{\tt n}]]]]$,
as matching $M$ against {\tt n+n+n}, {\tt n+n}, and {\tt n} all consume just the
first {\tt n} while generating $M[{\tt n}]$, because we have $\Matgm{M-n \ / \
  n}{{\tt n+n+n}}{\{(M, {\tt n+n+n}) \mapsto \Nothing\}} \Lp ({\tt +n+n},{\tt n})$ via {\bf
  lvar.3}, {\bf con.3}, and {\bf ord.3}, and we have $\Matgm{M-n \ / \
  n}{{\tt n+n+n}}{\{(M,{\tt n+n+n}) \mapsto ({\tt +n+n},{\tt n})\}} \Li ({\tt +n+n},{\tt n})$ via
{\bf inc.3}. The same holds for subjects {\tt n+n} and {\tt n} with
different suffixes.  Now, when $E$ matches {\tt n+n+n} we will have $M$ in
$M+E$ matching the first {\tt n}, with $E$ recursively matching the
second {\tt n+n}, with $M$ again matching the first {\tt n} and $E$
recursively matching the last {\tt n} via the second alternative.

Matching $E$ with {\tt n-n-n} yields $E[M[M[M[{\tt n}]{\tt -n}]{\tt -n}]]$, as $M$
now matches {\tt n-n-n} with a proof tree similar to our first example ($E
\leftarrow E+n \ / \ n$ against {\tt n+n+n}). The first alternative of $E$
fails because $M$ consumed the whole subject, and the second
alternative yields the final result via {\bf ord.3} and {\bf var.1}.

The semantics of Figure~\ref{fig:leftrecmemo} also handles indirect and
mutual left recursion well. The following mutually left-recursive PEG
is a direct translation of the CFG used as the last example of Section~2:
\begin{align*}
L & \leftarrow P.x \ / \ x \\
P & \leftarrow P\boldsymbol{(}n\boldsymbol{)} \ / \ L 
\end{align*}

It is instructive to work out what happens when matching $L$ with a
subject such as {\tt x(n)(n).x(n).x}. We will use our superscript
notation for bounded recursion, but it is easy to check that the
explanation corresponds exactly with what is happening with the
semantics using $\mathcal{L}$.

The first alternative of $L^1$ will
fail because both alternatives of $P^1$ fail, as they use $P^0$, due
to the direct left recursion on $P$, and $L^0$, due to the indirect
left recursion on $L$. The second alternative of $L^1$ matches the
first {\tt x} of the subject. Now $L^2$ will try to match $P^1$ again,
and the first alternative of $P^1$ fails because it uses $P^0$, while
the second alternative uses $L^1$ and matches the first {\tt x}, and
so $P^1$ now matches {\tt x}, and we have to try $P^2$, which will
match {\tt x(n)} through the first alternative, now using $P^1$. $P^3$
uses $P^2$ and matches {\tt x(n)(n)} with the first alternative, but
$P^4$ matches just $x$ again, so $P^3$ is the answer, and $L^2$
matches {\tt x(n)(n).x} via its first alternative.

$L^3$ will try to match $P^1$ again, but $P^1$ now matches {\tt
  x(n)(n).x} via its second alternative, as it uses $L^2$. This means
$P^2$ will match {\tt x(n)(n).x(n)}, while $P^3$ will match {\tt
  x(n)(n).x} again, so $P^2$ is the correct bound, and $L^3$ matches
{\tt x(n)(n).x(n).x}, the entire subject. It is easy to see that $L^4$
will match just $x$ again, as $P^1$ will now match the whole subject
using $L^3$, and the first alternative of $L^4$ will fail.

Intuitively, the mutual recursion is playing as nested repetitions,
with the inner repetition consuming {\tt (n)} and the outer repetition
consuming the result of the inner repetition plus {\tt .x}. The result
is a PEG equivalent to the PEG for l-values in the end of Section~2 in
the subjects it matches, but that yields parse strings that are
correctly left-associative on each {\tt (n)} and {\tt .x}.

We presented the new rules as extensions intended only for non-terminals
with left-recursive rules, but this is not necessary: the {\bf lvar}
rules can replace {\bf var} without changing the result of any proof
tree. If a non-terminal does not appear in a left-recursive position
then rules {\bf lvar.3} and {\bf lvar.4} can never apply by
definition. These rules are the only place in the semantics where the
contents of $\mathcal{L}$ affects the result, so {\bf lvar.2} is
equivalent to {\bf var.2} in the absence of left
recursion. Analogously, if $\Matgm{(P(A)}{xy}{\mathcal{L}[(A,xy) \mapsto
  \Nothing]} \Lp (y, x^\prime)$ then $\Matgm{(P(A)}{xy}{\mathcal{L}[(A,xy) \mapsto
  (y, x^\prime)]} \Lp (y, x^\prime)$ in the absence of left recursion,
so we will always have $\Matgm{A}{xy}{\mathcal{L}[(A,xy) \mapsto
  (y, x^\prime)]} \Li (y, x^\prime)$ via {\bf inc.3}, and {\bf lvar.1}
is equivalent to {\bf var.1}. We can formalize this argument with the
following lemma:
\begin{lemma}[Conservativeness]
Given a PEG $G$, a parsing expression $p$ and a subject $xy$, we have
one of the following: if $\Matg{p}{xy} \Lp \Anyx$, where $\Anyx$ is
$\Nothing$ or $(y, x^\prime)$, then $\Matgm{p}{xy}{\mathcal{L}} \Lp
\Anyx$, as long as $(A,w) \not\in \mathcal{L}$ for any non-terminal
$A$ and subject $w$ appearing as $\Matg{A}{w}$ 
in the proof tree of if $\Matg{p}{xy} \Lp \Anyx$.
\end{lemma}
\begin{proof}
By induction on the height of the proof tree for $\Matg{p}{xy} \Lp
\Anyx$. Most cases are trivial, as the extension of their rules with
$\mathcal{L}$ does not change the table. The interesting cases are
{\bf var.1} and {\bf var.2}.

For case {\bf var.2} we need to use rule {\bf lvar.2}. We introduce
$(A, xy) \mapsto \Nothing$ in $\mathcal{L}$, but $\Matg{A}{xy}$ cannot
appear in any part of the proof tree of $\Matg{P(A)}{xy} \Lp
\Nothing$, so we can just use the induction hypothesis.

For case {\bf var.1} we need to use rule {\bf lvar.1}. Again we have $(A,
xy) \mapsto \Nothing$ in $\mathcal{L}$, but we can use the induction
hypothesis on $\Matgm{P(A)}{xy}{\mathcal{L}[(A,xy) \mapsto \Nothing]}$
to get $(y, x^\prime)$. We also use {\bf inc.3} to get $\Matgm{P(A)}{xy}{\mathcal{L}[(A,xy) \mapsto (y,
  x^\prime) \Li (y, x^\prime)]}$ from $\Matgm{P(A)}{xy}{\mathcal{L}[(A,xy) \mapsto (y,
  x^\prime)]}$, using the induction hypothesis, finishing {\bf lvar.1}.
\end{proof}

In order to prove that our semantics for PEGs with left-recursion
gives meaning to any closed PEG (that is, any PEG $G$ where $P(A)$ is
defined for all non-terminals in $G$) we have to fix the case where a
repetition may not terminate ($p$ in $p^*$ has not failed but not
consumed any input). We can add a $x \neq \Epsi$ predicate to rule {\bf
  rep.2} and then add a new rule:
\begin{align*}
& {\frac{\Matgm{p}{x}{\mathcal{L}} \Lp \Tup{x}{\Epsi}}
	{\Matgm{p^*}{x}{\mathcal{L}} \Lp \Tup{x}{\Epsi}}} \mylabel{rep.3} 
\end{align*}

We also need a well-founded ordering $<$ among the elements of the left side
of relation $\Lp$. For the subject we can use $x < y$ if and only if
$x$ is a proper suffix of $y$ as the order, for the parsing expression
we can use $p_1 < p_2$ if and only if $p_1$ is a proper part of the
structure of $p_2$, and for $\mathcal{L}$ we can use $\mathcal{L}[A
\mapsto (x,y)] < \mathcal{L}$ if and only if either $\mathcal{L}(A)$ is not
defined or $x < z$, where $\mathcal{L}(A) = (z, w)$. Now we can prove
the following lemma:
\begin{lemma}[Completeness]
Given a closed PEG $G$, a parsing expression $p$, a subject $xy$, and a
memoization table $\mathcal{L}$, we have either
$\Matgm{p}{xy}{\mathcal{L}} \Lp (y,x^\prime)$ or
$\Matgm{p}{xy}{\mathcal{L}} \Lp \Nothing$.
\end{lemma}
\begin{proof}
By induction on the triple $(\mathcal{L},xy,p)$. It is straightforward
to check that we can always use the induction hypothesis on the
antecedent of the rules of our semantics.
\end{proof}

We can write expression grammars with several operators with different precedences
and associativities with the use of multiple non-terminals for the
different precedence levels, and the use of left and right recursion
for left and right associativity. A simple addition to the semantics,
though, lets us write these grammars in a style that is more familiar
to users of LR parsing tools. The next section presents this addition.

\section{Controlling Operator Precedence}
\label{sec:prec}

A non-obvious consequence of our bounded left recursion semantics is
that a rule that mixes left and right recursion is
right-associative. For example, matching $E \leftarrow E+E \ /\ n$
against {\tt n+n+n} yields the parse string $E[E[n]+E[E[n]+E[n]]]$,
where the {\tt +} operator is right-associative. The
reason is that $E^2$ already matches the whole string:
\begin{align*}
E^1 & \leftarrow E^0 + E \ / \ n = n \\
E^2 & \leftarrow E^1 + E \ / \ n = n + E \ / \ n
\end{align*}

We have the first alternative of $E^2$ matching {\tt n+} and then
trying to match $E$ with {\tt n+n}. Again we will have $E^2$ matching
the whole string, with the first alternative matching {\tt n+} and
then matching $E$ with {\tt n} via $E^1$. 

The algorithm proposed by Warth et al.
to handle left-recursion in PEGs~\cite{warth:left} also has a right-recursive
bias in rules that have both left and right-recursion, 
as noted by Tratt~\cite{tratt:left}. In this section, we give an extension
to our semantics that lets the user control the associativity instead of
defaulting to right-associativity in grammars like the above one.

Grammars that mix left and right recursion are also problematic
when parsing CFGs, as mixing left and right recursion leads to
ambiguous grammars. The CFG $E \rightarrow E+E \ |\ n$ is ambiguous,
and will cause a shift-reduce conflict in an LR parser.
Resolving the conflict by shifting (the default in LR parser
generators) also leads to right-associative operators.

Mixing operators, such as in
the grammar $E \rightarrow E+E \ |\ E*E \ |\ n$, will also cause
LR conflicts. Writing these grammars is convenient when the number of
operators is large, though, and an LR parser can control the
associativity and precedence among different operators by picking
which action to perform in the case of those shift-reduce
conflicts. LR parser generators turn a specification of the
precedence and associativity of each operator into a recipe for
resolving conflicts in a way that will yield the correct parses for
operator grammars that mix left and right recursion.

It is possible to adapt our left-recursive PEG semantics to also let
the user control precedence and associativity in PEGs such as $E
\leftarrow E+E \ /\ n$ and $E \leftarrow E+E \ /\ E*E \ /\ n$. The
idea is to attach a {\em precedence level} to each occurrence of a
left-recursive non-terminal, even if it is not in a left-recursive
position. The level of a non-terminal is stored in the memoization table $\mathcal{L}$,
and a left-recursive use of the non-terminal fails if its level is less
than the level stored in the table. These precedence levels match the way
left and right precedence numbers are assigned to operators in {\em precedence
climbing} parsers~\cite{hanson:prec}.

Figure~\ref{fig:leftrecprec} gives the semantics for left-recursive
non-terminals with precedence levels.  The default level is $1$.
The changes are minimal: rule {\bf lvar.1} now stores the attached
level in the memoization table before trying to increase the bound,
while {\bf lvar.4} guards against the attached level being lesser than
the stored level. A new rule {\bf lvar.5} fails the match of $A_k$ if
the attached level $k$ is less than the level stored in the memoization
table. Rules {\bf lvar.2}, {\bf lvar.3}, and {\bf inc.2} are unchanged,
and rules {\bf inc.1} and {\bf inc.3} just propagate the store level
without changes. These five rules have been elided.

Going back to our previous example, the grammar $E \leftarrow E+E \ /\ n$,
we can get a left-associative $+$
by giving the left-recursive $E$ a precedence level that is greater
than the precedence level of the right-recursive $E$. So we should
rewrite the previous grammar as $E \leftarrow E_1+E_2 \ /\ n$. 
As the precedence level of the right-recursive
$E$ is $2$, after a right-recursive call the left-recursive call $E_1$,
with precedence level $1$, always fails when trying to match. This
means that $E_2$ can match just a prefix {\tt n} of the subject,
via the second alternative of the ordered choice.
 
In a similar way, we could rewrite $E \leftarrow E+E \ /\ n$ 
as $E \leftarrow E_2+E_1 \ /\ n$, or just use the default level and 
interpret it as $E \leftarrow E_1+E_1 \ /\ n$,
to get a right-associative $+$ in our new semantics that has
precedence levels. Now after a right-recursive call the left-recursive
call can match an unlimited prefix of the subject. 

We could impose left-associativity as the default in case of a grammar that
mixes left and right recursion just by changing $k \ge k^\prime$ in 
rule {\bf lvar.4} to $k > k^\prime$ and by changing $k < k^\prime$ in rule 
{\bf lvar.5} to $k \le k^\prime$. The disadvantage is that this makes specifying
precedence levels with right-associative operators a little harder, but this is
an user interface issue: a tool could use an YACC-style interface, with
{\tt \%left} and {\tt \%right} precedence directives and the order these directives appear
giving the relative precedence, with the tool assigning precedence levels automatically.
We left right-associativity as the default for the semantics presented in this section
because it makes the multi-operator and multi-associativity grammar example we give next
easier to understand.

\begin{figure}[t]
{\small
\begin{align*}
& 
{\frac{
\begin{array}{c}
\;\;\;\;
(A,xyz) \notin \mathcal{L} \interf
    \Matgm{P(A)}{xyz}{\mathcal{L}[(A,xyz) \mapsto \Nothing]} \Lp
    (yz,x^\prime)  \;\;\;\; \\ \;\;\;\;
				\Matgm{P(A)}{xyz}{\mathcal{L}[(A,xyz)
                                  \mapsto (yz,x^\prime,k)]} \Li
                                (z,(xy)^\prime) \;\;\;\;
\end{array}
}
	{\Matgm{A_k}{xyz}{\mathcal{L}} \Lp (z,A[(xy)^\prime])}} \mylabel{lvar.1} \\\\
&
      \fivespaces {\frac{\mathcal{L}(A,xy) = (y,x^\prime,k^\prime)\ \
          \ k \ge k^\prime}
	{\Matgm{A_k}{xy}{\mathcal{L}} \Lp (y,A[x^\prime])}} \mylabel{lvar.4}
      \interf
       {\frac{\mathcal{L}(A,xy) = (y,x^\prime,k^\prime)\ \ \ k < k^\prime}
	{\Matgm{A_k}{xy}{\mathcal{L}} \Lp \Nothing}} \mylabel{lvar.5}
\end{align*}
\caption{Semantics for left-recursive non-terminals with precedence levels}
\label{fig:leftrecprec}
}
\end{figure}

In a grammar with several operators grouped in different precedence
classes, the precedence classes would map to different levels in the
non-terminals. A grammar with left-associative $+$ and $-$ operators
in the lowest precedence class, left-associative $*$ and $\div$
operators in the middle precedence class, a right-associative $**$
operator in second-to-highest precedence class, an unary $-$ in the
highest precedence class, and parentheses, could be described by the
following PEG:
\begin{align*}
E & \leftarrow E_1 + E_2 \ / \ E_1 - E_2 \ /\ E_2 * E_3 \ /\ E_2 \div
E_3 \ /\ E_3 *\!*\; E_3 \ / \ -E_4 \ / \ (E_1) \ / \ n 
\end{align*}

As an example of how precedence levels work in practice, let us
consider the PEG $E \leftarrow E_1+E_2 \ /\ E_2*E_2 \ /\ n$, where we
set precedence levels so a right-associative $*$ has precedence over a
left-associative $+$, and check what happens on the inputs {\tt
  n+n+n}, {\tt n*n*n}, {\tt n*n+n}, and {\tt n+n*n}. We will assume
that we are trying to match $E_1$ against each subject.

For {\tt n+n+n} we first try to match $E_1+E_2 \ /\ E_2*E_2 \ /\ n$
with $\mathcal{L}(E,\mathtt{n+n+n}) = \Nothing$. The first two
alternatives will fail via {\bf lvar.3}, and the third matches the
first {\tt n}. We then try to increase the bound with
$\mathcal{L}(E,\mathtt{n+n+n}) = (\mathtt{+n+n},n,1)$. In the first
alternative, $E_1$ matches via {\bf lvar.4}, and we start again with
$E_2$ by trying to match $E_1+E_2 \ /\ E_2*E_2 \ /\ n$ against {\tt
  n+n} with $\mathcal{L}(E,\mathtt{n+n}) = \Nothing$. 

Again, the third
alternative matches the first {\tt n}, and we will try to increase the
bound with $\mathcal{L}(E,\mathtt{n+n}) = (\mathtt{+n},n,2)$. Notice
the precedence level is $2$. This level makes the first alternative of $E_1+E_2 \ /\ E_2*E_2 \
/\ n$ fail against {\tt n+n}, and the second alternative fails on
$*$. The third alternative succeeds, but does not increase the bound,
so matching $E_2$ against {\tt n+n} will match just the first {\tt n},
and $E_1+E_2$ will match the first {\tt n+n} of the input {\tt
  n+n+n}. 

We then try to increase the bound of $E_1$ against {\tt n+n+n} again,
 with $\mathcal{L}(E,\mathtt{n+n+n}) = (\mathtt{+n},E[n]+E[n],1)$. The
 first alternative of $E$ will now match the whole subject, yielding
 the parse string $E[E[n]+E[n]]+E[n]$, and a third try to increase the bound
 will fail, and the final result is $E[E[E[n]+E[n]]+E[n]]$, with a
 left-associative $+$.

For the subject {\tt n*n*n}, after matching the first $n$, we try to
increase the bound with $\mathcal{L}(E,\mathtt{n*n*n}) =
(\mathtt{*n*n},n,1)$. The precedence level is still $1$ because we are
matching $E_1$ against the whole subject. The first alternative fails
on $+$, the first $E_2$ of the second alternative succeeds via {\bf
  lvar.4}, and we will try the second $E_2$ of the second alternative
against {\tt n*n}. The first {\tt n} matches with
$\mathcal{L}(E,\mathtt{n*n}) = \Nothing$, and then we try to increase
the bound with $\mathcal{L}(E,\mathtt{n*n}) = (\mathtt{*n},n,2)$. Now
the second alternative succeeds again with {\bf lvar.4}, as the
precedence level is $2$, and eventually the final result is
$E[E[n]*E[E[n]*E[n]]]$, with a right-associative $*$.

The behavior on the other two subjects, {\tt n*n+n} and {\tt n+n*n},
are respectively analogous to the first two, with the right-recursive
$E_2$ on $E_2*E_2$ failing to match {\tt n+n} on the first subject
because the level is $2$, and the right-recursive $E_2$ on $E_1+E_2$
matching {\tt n*n} on the second subject due to the same reason,
yielding the results $E[E[E[n]*E[n]]+E[n]]$ and
$E[E[n]+E[E[n]*E[n]]]$. Both results have the correct precedence.

While it is natural to view a PEG as a blueprint for a recursive
descent parser with local backtracking, with optional use of
memoization for better time complexity, an alternative execution
model compiles PEGs to instructions on a lower-level {\em parsing
  machine}~\cite{dls:lpeg}. The next section briefly describes the
formal model of one such machine, and extends it to have the same
behavior on left-recursive PEGs as the semantics of
Figures~\ref{fig:leftrecmemo} and~\ref{fig:leftrecprec}.

\section{A Parsing Machine for Left-recursive PEGs}
\label{sec:machine}

The core of LPEG, an implementation of PEGs for the Lua
language~\cite{roberto:lpeg}, is a virtual parsing machine. LPEG
compiles each parsing expression to a program of the parsing machine,
and the program for a compound expression is a combination of the
programs for its subexpressions. The parsing machine has a formal
model of its operation, and proofs that compiling a PEG yields a
program that is equivalent to the original PEG~\cite{dls:lpeg}.

The parsing machine has a register to hold the program counter used to
address the next instruction to execute, a register to hold the
subject, and a stack that the machine uses for pushing call frames and
backtrack frames. A call frame is just a return address for the
program counter, and a backtrack frame is an address for the program
counter and a subject. The machine's instructions manipulate the
program counter, subject, and stack.

The compilation of the following PEG (for a sequence of zero or more
$a$s followed by any character other than $a$ followed by a $b$) uses
all of the basic instructions of the parsing machine:
\begin{eqnarray*}
A & \rightarrow & !a \, . \, B \ / \ aA \\
B & \rightarrow & b
\end{eqnarray*}

This PEG compiles to the following program:
\begin{eqnarray*}
& &    \In{Call}{A} \\
A\!\!: & & \In{Choice}{A1} \\
& &    \In{Choice}{A2} \\
& &    \In{Char}{a} \\
& &    \In{Commit}{A3} \\
A3\!\!: & & \In{Fail}{} \\
A2\!\!: & & \In{Any}{} \\
& &    \In{Call}{B} \\
& &    \In{Commit}{A4} \\
A1\!\!: & & \In{Char}{a} \\
& &    \In{Jump}{A} \\
A4\!\!: & & \In{Return}{} \\
B\!\!: & & \In{Char}{b} \\
& &    \In{Return}{}
\end{eqnarray*}

The behavior of each instruction is straightforward: \texttt{Call}
pushes a call frame with the address of the next instruction and jumps
to a label, \texttt{Return} pops a call frame and jumps to its
address, \texttt{Jump} is an unconditional jump, \texttt{Char} tries
to match a character with the start of the subject, consuming the
first character of the subject if successful, \texttt{Any} consumes
the first character of the subject (failing if the subject is
$\Epsi$), \texttt{Choice} pushes a backtrack frame with the subject
and the address of the label, \texttt{Commit} discards the backtrack
frame in the top of the stack and jumps, and \texttt{Fail} forces a
failure. When the machine enters a failure state it pops call frames
from the stack until reaching a backtrack frame, then pops this frame
and resumes execution with the subject and address stored in it.

Formally, the program counter register, the subject register, and
the stack form a machine state. We represent it as a tuple
$\mathbb{N} \times \mathrm{T}^* \times \mathrm{Stack}$, 
in the order above. A machine state can also be a
failure state, represented by $\Sfail{e}$, where $e$ is the
stack. Stacks are lists of $(\mathbb{N} \times \mathrm{T}^*)
\,\cup\, \mathbb{N}$, where
$\mathbb{N} \times \mathrm{T}^*$ represents a 
backtrack frame and $\mathbb{N}$ represents a call frame.

Figure~\ref{fig:semantics} presents the operational semantics of the 
parsing machine as a relation between machine states. The program
$\mathcal{P}$ that the machine executes is implicit. The relation
$\xrightarrow{\mathrm{Instruction}}$ relates two states when \emph{pc}
in the first state addresses a instruction matching the label, and the
guard (if present) is valid.

\begin{figure}[t]
{\small
\begin{eqnarray*}
\Sstep{\Sstate{pc}{a\!:\!s}{e}}{\Ia{Char}{a}}%
      {\Sstate{pc+1}{s}{e}}{}
\Sstep{\Sstate{pc}{b\!:\!s}{e}}{\Ia{Char}{a}}%
      {\Sfail{e}}{, a \neq b}
\Sstep{\Sstate{pc}{\Epsi}{e}}{\Ia{Char}{a}}%
      {\Sfail{e}}{}
\Sstep{\Sstate{pc}{a\!:\!s}{e}}{\Ia{Any}{}}%
      {\Sstate{pc+1}{s}{e}}{}
\Sstep{\Sstate{pc}{\Epsi}{e}}{\Ia{Any}{}}%
      {\Sfail{e}}{}
\Sstep{\Sstate{pc}{s}{e}}{\Ia{Choice}{i}}%
      {\Sstate{pc+1}{s}{(pc+i,s)\!:\!e}}{}
\Sstep{\Sstate{pc}{s}{e}}{\Ia{Jump}{l}}%
      {\Sstate{pc+l}{s}{e}}{}
\Sstep{\Sstate{pc}{s}{e}}{\Ia{Call}{l}}%
      {\Sstate{pc+l}{s}{(pc+1)\!:\!e}}{}
\Sstep{\Sstate{pc_1}{s}{pc_2\!:\!e}}{\Ia{Return}{}}%
      {\Sstate{pc_2}{s}{e}}{}
\Sstep{\Sstate{pc}{s}{h\!:\!e}}{\Ia{Commit}{l}}%
      {\Sstate{pc+l}{s}{e}}{}
\Sstep{\Sstate{pc}{s}{e}}{\Ia{Fail}{}}%
      {\Sfail{e}}{}
\Ssteppp{\Sfail{pc\!:\!e}}%
      {\Sfail{e}}{}
\Sstepppf{\Sfail{(pc,s)\!:\!e}}%
      {\Sstate{pc}{s}{e}}{}
\end{eqnarray*}
\caption{Operational semantics of the parsing machine}
\label{fig:semantics}
}
\end{figure}

We can extend the parsing machine to handle left recursion with
precedence levels by changing the semantics of the {\tt Call} and {\tt
  Return} instructions, as well as the semantics of failure
states. The new {\tt Call} instruction also needs a second parameter, the
precedence level. We also need to add a new kind of frame to the
machine stack: a {\em left-recursive call} frame. 

A left-recursive call frame to a non-terminal $A$ is a tuple
$(pc_r,pc_A,s,X,k)$ where $pc_r$ is the address of the instruction
that follows the call instruction, $pc_A$ is the address of the first
instruction of the non-terminal $A$, $s$ is the subject at the call site,
$X$ is either $\Nothing$ or a suffix of $s$, and is the memoized
result, and $k$ is the precedence level. These frames encode both the
control information for the machine, and the memoization table for
left recursion.

The idea of the extension is for the {\tt Call} instruction to handle
the {\bf lvar.3}, {\bf lvar.4}, and {\bf lvar.5} cases, as well as
parts of {\bf lvar.1} and {\bf lvar.2}. {\tt Return} then handles the
rest of {\bf lvar.1}, as well as {\bf inc.1} and {\bf inc.3}, while
the semantics of the failure states handle {\bf inc.2} and the rest of
{\bf lvar.2}.

Figure~\ref{fig:machineleft} gives the operational semantics of the
extension, and replace the rules of Figure~\ref{fig:semantics} that
deal with {\tt Call}, {\tt Return}, and failure states. As a
convenience, we use the notation $(pc,s) \notin \mathcal{L}_e$ to mean that there is
no left-recursive call frame in $e$ with $pc$ in the second position and $s$
in the third position, the notation $\mathcal{L}_e(pc,s) = \Nothing$ to
mean that in the left-recursive call frame of $e$ with $pc$ in the second
position and $s$ in the third position, the value in the fourth
position is $\Nothing$, and the notation $\mathcal{L}_e(pc,s) =
(s^\prime,k)$ to mean that in the left-recursive call frame of $e$
with $pc$ in the second position and $s$ in the third position, the
values of the fourth and fifth positions are $s^\prime$ and $k$, respectively.

\begin{figure}[t]
{\small
\begin{eqnarray*}
\Sstep{\Sstate{pc}{s}{e}}{\Iaa{Call}{l}{k}}%
      {\Sstate{pc+l}{s}{(pc+1,pc+l,s,\Nothing,k)\!:\!e}}{,\ \mathrm{where}\ (pc+l,s) \notin \mathcal{L}_e}
\Sstep{\Sstate{pc}{s}{e}}{\Iaa{Call}{l}{k}}%
       {\Sfail{e}}{,\ \mathrm{where}\ \mathcal{L}_e(pc+l,s) = \Nothing}
\Sstep{\Sstate{pc}{s}{e}}{\Iaa{Call}{l}{k}}%
      {\Sstate{pc+1}{s^\prime}{e}}{,\ \mathrm{where}\ \mathcal{L}_e(pc+l,s) = (s^\prime,
         k^\prime)\ \mathrm{and}\ k \ge k^\prime}
\Sstep{\Sstate{pc}{s}{e}}{\Iaa{Call}{l}{k}}%
       {\Sfail{e}}{,\ \mathrm{where}\ \mathcal{L}_e(pc+l,s) = (s^\prime,
         k^\prime)\ \mathrm{and}\ k < k^\prime}
\Sstep{\Sstate{pc}{s^{\prime\prime}}{(pc_r,pc_A,s,s^\prime,k)\!:\!e}}{\Ia{Return}{}}%
      {\Sstate{pc_A}{s}{(pc_r,pc_A,s,s^{\prime\prime},k):e}}{,\
        \mathrm{where}\ |s^{\prime\prime}| < |s^\prime|\ \mathrm{or}\ s^\prime =
      \Nothing}
\Sstep{\Sstate{pc}{s^{\prime\prime}}{(pc_r,pc_A,s,s^\prime,k)\!:\!e}}{\Ia{Return}{}}%
      {\Sstate{pc_r}{s^\prime}{e}}{,\
        \mathrm{where}\ |s^{\prime\prime}| >= |s^\prime|}
\Ssteppp{\Sfail{(pc_r,pc_A,s,s^\prime,k)\!:\!e}}%
      {\Sstate{pc_r}{s^\prime}{e}}{}
\Ssteppp{\Sfail{pc\!:\!e}}%
      {\Sfail{e}}{}
\Sstepppf{\Sfail{(pc,s)\!:\!e}}%
      {\Sstate{pc}{s}{e}}{}
\end{eqnarray*}
\caption{Semantics of left recursion in the parsing machine}
\label{fig:machineleft}
}
\end{figure}

The formal model of the parsing machine
represents the compilation process using a transformation function
$\mathrm{\Pi}$. The term $\mathrm{\Pi}(G, i, p)$ is the translation of
pattern $p$ in the context of the grammar $G$, where $i$ is the
position where the program starts relative to the start of the
compiled grammar. We use the notation $|\mathrm{\Pi}(G, i, p)|$
to mean the number of instructions in the program $\mathrm{\Pi}(G, i,
p)$.

The following lemma gives the correctness condition for the
transformation $\Pi$ for non-left-recursive PEGs in the original
semantics of PEGs and the parsing machine, based on an extension of the relation
$\xrightarrow{\mathrm{Instruction}}$ for a concatenation of multiple instructions:
\begin{lemma}[Correctness of $\Pi$ without left recursion]
Given a PEG $G$, a parsing expression $p$, and a subject $xy$, 
if $\Matg{p}{xy} \Lp (y,x^\prime)$ then $\Sstate{pc}{xy}{e}
\xrightarrow{\Pi(G,i,p)} \Sstate{pc+|\Pi(G,i,p)|}{y}{e}$, and if $\Matg{p}{xy} \Lp \Nothing$
then $\Sstate{pc}{xy}{e} \xrightarrow{\Pi(G,i,p)} \Sfail{e}$, where
$pc$ is the address of the first instruction of $\Pi(G,i,p)$.
\end{lemma}
\begin{proof}
Given in the paper that describes the parsing machine~\cite{dls:lpeg}.
\end{proof}

Before proving a revised correctness lemma for the compilation of
left-recursive PEGs, we need to relate memoization tables
$\mathcal{L}$ with stacks that have left-recursive call frames. We say
that a stack $e$ is {\em consistent} with a memoization table
$\mathcal{L}$ if and only if, for every entry in $\mathcal{L}$ with a
non-terminal $A$, a subject $s$, a result $X$ (either a suffix of $s$
or $\Nothing$), and a precedence level $k$, there is a left-recursive
call frame $(pc_r,pc_A,s,X,k)$ in $e$ where $pc_A$ is the position
where the program for $P(A)$ starts, and $pc_R$ is any position.

Intuitively, an stack $e$ is consistent with a memoization table
$\mathcal{L}$ if each entry in $\mathcal{L}$ has a corresponding
left-recursive call frame in $e$, and vice-versa.

The lemma below revises the correctness condition for the
transformation $\Pi$ to account for both the left-recursive PEG
semantics and the left-recursive parsing machine semantics:
\begin{lemma}[Correctness of $\Pi$ with left recursion]
Given a PEG $G$, a parsing expression $p$, a non-terminal $A$ of $G$, a subject $xyz$, a
memoization table $\mathcal{L}$, and a stack $e$ consistent with $\mathcal{L}$,
if $\Matgm{p}{xyz}{\mathcal{L}} \Lp (z,w)$ then $\Sstate{pc}{xyz}{e}
\xrightarrow{\Pi(G,i,p)} \Sstate{pc+|\Pi(G,i,p)|}{z}{e}$, if $\Matgm{p}{xyz}{\mathcal{L}} \Lp \Nothing$
then $\Sstate{pc}{xyz}{e} \xrightarrow{\Pi(G,i,p)} \Sfail{e}$, and if
$\Matgm{P(A)}{xyz}{\mathcal{L}\ \cup\ \{(A,xyz) \mapsto (yz,v)\}} \Li
(z,w)$ then $\Sstate{pc}{xyz}{(pc_r,pc_A,xyz,yz,k):e}
\xrightarrow{\Pi(G,i,P(A))\ \mathrm{Return}} \Sstate{pc_r}{y}{e}$.
\end{lemma}
\begin{proof}
By induction on the heights of the proof trees for the
antecedents. Most cases are similar to the ones in the proof of the
previous lemma~\cite{dls:lpeg}, as they involve parts of the semantics
of PEGs and the parsing machine that are unchanged. The interesting
cases occur when $p$ is a non-terminal $A$ (in the case of the
relation $\Lp$) or the right side of the production of a non-terminal
$A$ (in the case of relation $\Li$).

The subcases where the proof tree of $A$ ends with {\bf lvar.3}, {\bf
  lvar.4}, or {\bf lvar.5} follow from the consistency constraint on
the stack $e$ and the corresponding rule for the {\tt Call}
instruction. Subcase {\bf lvar.2} follows from an application of the
induction hypothesis, and the rule for failure states with a
left-recursive call frame on the top of the stack.

Subcase {\bf lvar.1} follows from the induction hypothesis on the
subtree ending with relation $\Lp$, the first transition rule for {\tt
  Return}, and the induction hypothesis on the subtree ending with
relation $\Li$.

Subcase {\bf inc.1} is similar to {\bf lvar.1}, while subcase {\bf
  inc.2} is similar to {\bf lvar.3}. Subcase {\bf inc.3} follows from
the induction hypothesis and the second transition rule for {\tt Return}.
\end{proof}

The use of left recursion is common in context free grammar definitions, and
several attempts have been made to give meaning to left-recursive PEGs
that would make it easier to reuse left-recursive CFGs. The next
section reviews previous work on left recursion for PEGs, and
related work on left recursion for top-down parsing approaches.

\section{Related Work}
\label{sec:related}

Warth et al.~\cite{warth:left} describes a modification of the packrat
parsing algorithm to support both direct and indirect left
recursion. The algorithm uses the packrat memoization table to detect
left recursion, and then begins an iterative process that is similar
to the process of finding the correct bound in our semantics.

Warth et al.'s algorithm is tightly coupled to the packrat parsing
approach, and its full version, with support for indirect left
recursion, is complex, as noted by the
authors~\cite{ometa:complexity1}. The released versions of the authors'
PEG parsing library, OMeta~\cite{warth:ometa}, only implement support
for direct left recursion to avoid the extra complexity~\cite{ometa:complexity1}.
Our approach does not rely on memoization; although the left recursion table $\mathcal{L}$
looks like a memoization table, the recursive structure of the natural semantics
makes it work more like a stack. The translation of our semantics to a parsing
machine in Section~5 makes this point clear.

Warth et al.'s algorithm also produces surprising results with some grammars,
both directly and indirectly left-recursive, that are related to its reuse of the
memoization table, as reported by some users
of Warth et al.'s algorithm~\cite{peglist:problemleft1, peglist:problemleft2, ometalist:problemleft1}.
As an example, consider the following grammar, based on the report of an user~\cite{peglist:problemleft1}:
\begin{align*}
S & \leftarrow X  \\
X & \leftarrow X \; Y \ / \ \Epsi \\
Y & \leftarrow x
\end{align*}

The algorithm fails to parse strings of the form $x^+$, because the
nullable left recursion of $X$ will make the algorithm mistakenly fail the
matching of  the non-terminal $Y$ that follows $X$ (if the first alternative
was $X \; x$ instead of $X \; Y$ the grammar would work).

Our semantics does not share these
issues, although it shows that a left-recursive packrat parser cannot
index the packrat memoization table just by a parsing expression and a
subject, as the $\mathcal{L}$ table is also involved. One solution to
this issue is to have a scoped packrat memoization table, with a new
entry to $\mathcal{L}$ introducing a new scope. We believe this
solution is simpler to implement in a packrat parser than fixing and implementing
Warth et al.'s close to fifty lines of pseudo-code (see the discussion about IronMeta below).

Tratt~\cite{tratt:left} presents an algorithm for supporting direct
left recursion in PEGs, based on Warth et al.'s, that does not use a
packrat memoization table and does not assume a packrat parser. 
Tratt's algorithm fixes the issues that Warth et al.'s algorithm has with nullable left
recursion, but only supports direct left recursion, while our solution works
with any kind of left recursion.

Tratt also presents a more complex algorithm that tries to ``fix'' the right-recursive bias in
productions that have both left and right recursion, like the $E
\leftarrow E+E \ /\ n$ example we discussed at the end of
Section~3. Although we do not believe this bias is a problem, our semantics
can give a more general solution to operator precedence and associativity in PEGs, 
as we have shown in Section~\ref{sec:prec}.

IronMeta~\cite{ironmeta} is a PEG library for the Microsoft Common Language Runtime,
based on OMeta~\cite{warth:ometa}, that supports direct and indirect
left recursion using an implementation of an unpublished preliminary
version of our semantics. This preliminary version is essentially the same, apart
from notational details and the presence of precedence levels, so IronMeta can be considered a working
implementation of the semantics of Figure~\ref{fig:leftrecmemo}. Initial versions of IronMeta used
Warth et al.'s algorithm for left recursion~\cite{warth:left}, but in
version 2.0 the author switched to an implementation of our semantics,
which he considered ``much simpler and more general''~\cite{ironmeta}.

Parser combinators~\cite{hutton:comb} are a top-down parsing method
that is similar to PEGs, being another way to declaratively specify a
recursive descent parser for a language, and share with PEGs the same issues of
non-termination in the presence of left recursion. Frost et
al.~\cite{frost:left} describes an approach for supporting left
recursion in parser combinators where a count of the number of
left-recursive uses of a non-terminal is kept, and the non-terminal
fails if the count exceeds the number of tokens of the input. We have
shown in Section~3 that such an approach would not work with PEGs,
because of the semantics of ordered choice (parser combinators use
the same non-deterministic choice operator as
CFGs). Ridge~\cite{ridge:allcfg} presents another way of implementing
the same approach for handling left recursion, and has the same issues
regarding its application to PEGs.

ANTLR~\cite{parr:antlrpldi} is a popular
parser generator that produces top-down parsers for Context-Free
Grammars based on LL(*), an extension of LL(k) parsing. Version 4 of
ANTLR will have support for direct left recursion that is specialized
for expression parsers~\cite{antlr:left}, handling precedence and associativity by
rewriting the grammar to encode a {\em precedence climbing}
parser~\cite{hanson:prec}. This support is heavily dependent on ANTLR
extensions such as semantic predicates and backtracking, and not translatable to PEGs.

Generalized LL parsing~\cite{scott:glltree,scott:gll} solves the problem of
left recursion in LL(1) recursive descent parsers, along with the problem of
LL(1) conflicts, by making the parse proceed among the different conflicting
alternatives ``in parallel'', replacing the recursive descent stack with a
{\em graph structured stack} (GSS), a data structure adapted from Generalized LR
parsers~\cite{tomita:gss}. The resulting parsing technique can parse the language of
any context-free language. As this includes ambiguous grammars, the tree construction
part of the parser is non-trivial to implement efficiently, using {\em shared packed parse forest}
data structure also adapted from generalized bottom-up techniques~\cite{scott:brnglr}.
 The behavior of the algorithm is tied to the non-deterministic
semantics of CFG alternatives, and it relies on consuming the whole input to
detect a successful parse, making the technique not applicable as a solution for left recursion
in PEGs, which have deterministic, ordered choice, and can successfully consume just
a prefix of the input.

\section{Conclusion}
\label{sec:con}

We presented a conservative extension to the semantics of PEGs that
gives an useful meaning for PEGs with left-recursive rules. It is
the first extension that is not based on packrat parsing as the
parsing approach, while supporting both direct and indirect left
recursion. The extension is based on bounded left recursion, where we
limit the number of left-recursive uses a non-terminal may have,
guaranteeing termination, and we use an iterative process to find the
smallest bound that gives the longest match for a particular use of
the non-terminal.

We also presented some examples that show how grammar writers can use
our extension to express in PEGs common left-recursive idioms from
Context-Free Grammars, such as using left recursion for
left-associative repetition in expression grammars, and the use of
mutual left recursion for nested left-associative repetition. We
augmented the semantics with {\em parse strings} to show how we get a
similar structure with left-recursive PEGs that we get with the parse
trees of left-recursive CFGs.

We have proved the conservativeness of our extension,
and also proved that all PEGs are complete with the extension, so
termination is guaranteed for the parsing of any subject with any PEG,
removing the need for any static checks of well-formedness beyond the
simple check that every non-terminal in the grammar has a rule.

We have also described a simple addition to our semantics that makes
it easier to describe expression grammars with multiple levels of
operator precedence and associativity, a pattern that users of LR
parser generators are used to.

Finally, we presented a semantics for describing left-recursive PEGs
as programs in a low-level {\em parsing machine}, as an extension to
the semantics of a parsing machine for non-left-recursive
PEGs~\cite{dls:lpeg}. We prove that this extension is correct with
regards to our semantics for left-recursive PEGs.

Our semantics has already been implemented in a PEG library that uses
packrat parsing~\cite{ironmeta}. We also have a prototype of the
left-recursive parsing machine semantics built on top of
LPEG~\cite{roberto:lpeg}, an implementation of the parsing machine for
regular PEGs that we extend. Preliminary benchmarks show that
left-recursive grammars perform similarly to the same grammars with
left recursion removed.

\bibliographystyle{model3a-num-names}
\bibliography{leftpeg}

\end{document}